\documentclass[journal]{IEEEtran}
\usepackage{amsmath, amssymb, amsthm, bm}
\usepackage{graphicx}
\usepackage{cite}
\usepackage{xcolor}
\usepackage{balance}

\newtheorem{theorem}{Theorem}
\newtheorem{corollary}{Corollary}

\newcommand{\E}{\mathbb{E}}
\newcommand{\CN}{\mathcal{CN}}
\newcommand{\Real}{\mathrm{Re}}

\title{Distortion Is Not Noise: On the Limits of the Kappa Model for Monostatic ISAC}

\author{Haofan~Dong,~\IEEEmembership{Student~Member,~IEEE,}
        and~Ozgur~B.~Akan,~\IEEEmembership{Fellow,~IEEE}%
\thanks{H.~Dong, and O.~B.~Akan are with the Internet of Everything Group,
Department of Engineering, University of Cambridge, Cambridge CB3~0FA, U.K.
(e-mail: hd489@cam.ac.uk; oba21@cam.ac.uk).}%
\thanks{Ozgur B. Akan is also with the Center for neXt-generation Communications
(CXC), Department of Electrical and Electronics Engineering,
Ko\c{c} University, 34450 Istanbul, Turkey.}}

\begin{document}
\maketitle

\begin{abstract}
Monostatic ISAC sensing differs from communication because the transmitter can monitor its distorted transmit waveform. 
Thus, the aggregate $\kappa$ distortion model, which treats impairments as unknown noise, is appropriate for communication but pessimistic for monostatic sensing. 
We derive PA-aware sensing Cram\'er--Rao bounds (CRBs) and a PN-aware CRB that reveals an irreducible velocity-error floor, and quantify when $\kappa$-based bounds overestimate sensing degradation. 
Simulations validate the analysis and show robustness to practical DPD template errors (less than 1~dB overhead at a typical $-25$~dB NMSE).
\end{abstract}


\section{Introduction}
\label{sec:intro}

Integrated sensing and communication (ISAC) has emerged as a key enabler for next-generation wireless systems~\cite{Liu2022survey, Sturm2011, Koivunen2024}.
A critical practical challenge is the impact of radio-frequency (RF) hardware impairments---particularly power amplifier (PA) nonlinearity and oscillator phase noise (PN)---on joint performance.

The dominant approach models all impairments via the aggregate $\kappa$ distortion model~\cite{Bjornson2013, Schenk2008, Studer2010}: $z = \sqrt{1{-}\kappa_t^2}\, x + \eta_t$ with $\eta_t \sim \CN(0, \kappa_t^2 P_X)$, adopted for ISAC analysis~\cite{Salman2025} and waveform design~\cite{MateosRamos2024}.
Notably,~\cite{Salman2025} directly applies the $\kappa$ model to derive sensing SCNR for ISAC precoding, and~\cite{MateosRamos2024} uses it as the HWI model for end-to-end ISAC learning---both inheriting the ``distortion-as-noise'' assumption for the sensing echo.
However, the $\kappa$ model assumes the distortion $\eta_t$ is \emph{unknown} at the receiver.
While valid for communication, this is violated in monostatic sensing, where the co-located transmitter knows its own distorted waveform (Fig.~\ref{fig:system}).

On the physics-based side, Keskin~\emph{et al.}~\cite{Keskin2023} derive the PN-aware sensing CRB via the augmented FIM for monostatic OFDM, and our prior work~\cite{Dong2025mimo} models individual HWI for THz ISAC.
However,~\cite{Keskin2023} does not address PA nonlinearity, does not compare against the $\kappa$ model, and does not establish the separability of PA and PN effects.
No prior work systematically compares the $\kappa$ model against physics-based HWI models for \emph{both} sensing and communication.

The contributions are: (i)~a physics-based sensing CRB under PA nonlinearity (Theorem~\ref{thm:pa_crb}), with SNR-independent degradation $<$1.1~dB, and a proof that the $\kappa$ model overestimates by a factor growing linearly with SNR (Theorem~\ref{thm:gap}); (ii)~a PN-aware Doppler CRB following the augmented-FIM approach in~\cite{Keskin2023} (Theorem~\ref{thm:pn}), revealing a velocity floor invisible to $\kappa$; and (iii)~a proof that the joint PA+PN CRB separates into an orthogonal (IBO, $\beta$) design space (Corollary~\ref{cor:joint})---absent from both~\cite{Keskin2023} and~\cite{Dong2025mimo}.
All results are validated by Monte Carlo simulation and shown robust to practical DPD template errors.

\emph{Notation:} Bold lowercase and uppercase denote vectors and matrices; $(\cdot)^H$ denotes Hermitian transpose; $\bm{I}_M$ is the $M\times M$ identity.\footnote{Additional notation: $(\cdot)^*$, $(\cdot)^T$ denote conjugate and transpose; $\Real\{\cdot\}$ denotes the real part.}

\section{System Model}
\label{sec:model}

\subsection{OFDM-ISAC Signal Model}

Consider a monostatic OFDM-ISAC system with $N$ subcarriers at spacing $\Delta f$ and $M$ symbols per coherent processing interval (CPI), with symbol duration $T = 1/\Delta f + T_{\mathrm{cp}}$ (Fig.~\ref{fig:system}).
The transmitter communicates with a single-antenna user and senses a point target with delay~$\tau$, Doppler~$\nu$, and reflectivity~$\alpha_s$.
The frequency-domain symbol $X[k,m]$ is drawn from a QAM constellation with energy $P_X = \E[|X[k,m]|^2]$.
We define centered indices $k' = k - (N{-}1)/2$ and $m' = m - (M{-}1)/2$, and the per-subcarrier sensing SNR as $\gamma_s \triangleq |\alpha_s|^2 P_X / \sigma^2$, where $\sigma^2$ is the noise variance.
The per-symbol aggregated sensing SNR used in the PN analysis is $\gamma_0 = 2 N\gamma_s$ (the factor~2 arises from parameterizing $\alpha_s = \alpha_R + j\alpha_I$ as two real unknowns in the FIM).
The target radial velocity is $v = \nu c/(2f_c)$ due to the monostatic two-way propagation.

\begin{figure}[!t]
\centering
\includegraphics[width=\columnwidth]{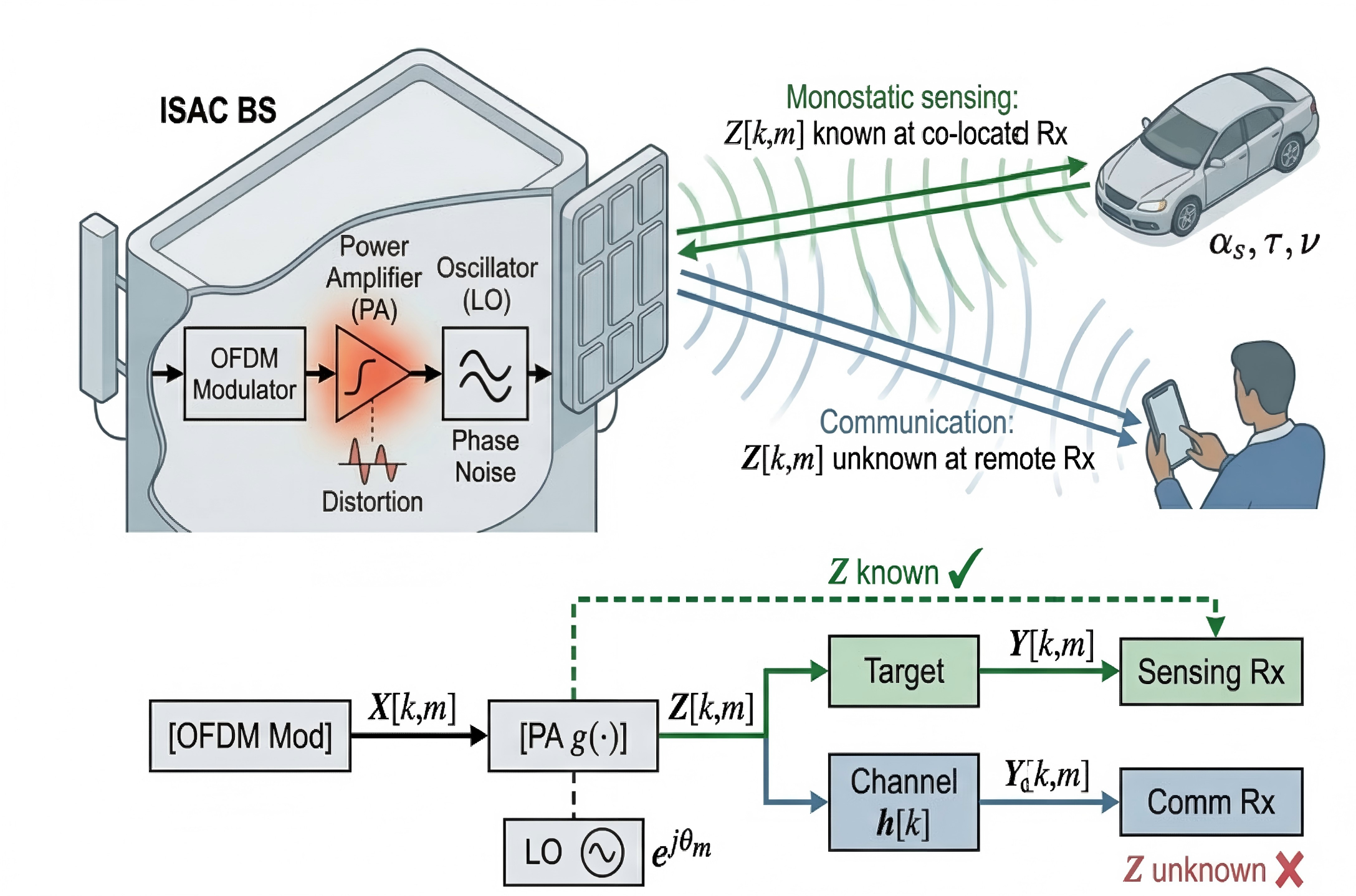}
\caption{Monostatic OFDM-ISAC system. \emph{Upper:} The ISAC base station transmits PA-distorted, phase-noise-affected waveforms for simultaneous sensing and communication. \emph{Lower:} Signal flow showing that $Z[k,m]$ is deterministically known at the co-located sensing receiver (\textcolor{green!60!black}{\checkmark}) but unknown at the remote communication receiver (\textcolor{red}{$\times$})---the asymmetry that the $\kappa$ model fails to capture.}
\label{fig:system}
\end{figure}

\subsection{PA Nonlinearity}

The frequency-domain symbols are converted to time domain, passed through a memoryless PA, and transformed back:
$x_m[n] = \frac{1}{\sqrt{N}}\sum_k X[k,m]\, e^{j2\pi kn/N}$,\\
$z_m[n] = g(x_m[n])$,
where $g(\cdot)$ is the PA input--output characteristic.
We adopt the Rapp model~\cite{Rapp1991}: $g(x) = x\bigl(1 + (|x|/A_{\mathrm{sat}})^{2p}\bigr)^{-1/(2p)}$ with saturation amplitude $A_{\mathrm{sat}}$ and smoothness factor~$p$.
The input back-off is $\mathrm{IBO} = 10\log_{10}(A_{\mathrm{sat}}^2/P_X)$~dB.
The PA output in frequency domain is $Z[k,m] = \frac{1}{\sqrt{N}}\sum_n z_m[n]\, e^{-j2\pi kn/N}$.

Crucially, since the transmitter generates $X[k,m]$ and applies $g(\cdot)$ digitally (or monitors the PA output via a digital pre-distortion (DPD) feedback path~\cite{Ghannouchi2009}), $Z[k,m]$ is a \emph{deterministic, known} quantity at the sensing receiver---it is not noise.
Theorem~\ref{thm:pa_crb} holds for any memoryless PA model; the Rapp model is adopted only for numerical evaluation.

\subsection{Phase Noise}

The oscillator introduces a common phase error (CPE) per OFDM symbol~\cite{Keskin2023, Wu2004}:
\begin{equation}
    \theta_m \approx \varphi(mT {+} T/2), \quad \bm{\theta} \sim \mathcal{N}(\bm{0}, \bm{C}_\theta),
    \label{eq:cpe}
\end{equation}
where $\varphi(t)$ is a Wiener process with 3-dB linewidth $\beta$, and $[\bm{C}_\theta]_{i,j} = \sigma_\Delta^2 \cdot \min(i{+}1, j{+}1)$ with $\sigma_\Delta^2 = 4\pi\beta T$.
The CPE approximation is accurate when $\sigma_\Delta^2 \ll 1$; for our parameters ($\beta \leq 1$~kHz, $T \approx 8.9$~$\mu$s), $\sigma_\Delta^2 \leq 0.12$, introducing $<$5\% error at the highest linewidth. The ICI-aware model of~\cite{Keskin2023} can be adopted for higher linewidths, where inter-subcarrier coupling modifies $J_{\tau,\theta_m}$ and may weaken the delay--PN decoupling.

Following~\cite{Keskin2023}, $\bm{\theta}$ represents accumulated phase from CPI start (worst case); shared-oscillator/short-range systems see partial TX--RX cancellation.

\subsection{Observation Models}

\emph{Sensing:} Under both PA and PN, the monostatic observation is
\begin{equation}
    Y[k,m] = \alpha_s\, e^{-j2\pi k'\Delta f \tau}\, e^{j2\pi\nu m'T}\, e^{j\theta_m}\, Z[k,m] + W[k,m],
    \label{eq:obs_sensing}
\end{equation}
where $W[k,m] \sim \CN(0, \sigma^2)$ is AWGN.

\emph{Communication:} The user receives
$Y_c[k,m] = h[k]\, e^{j\theta_m}\, Z[k,m] + W_c[k,m]$,
where $h[k]$ is the channel frequency response and $W_c \sim \CN(0, \sigma_c^2)$.
Unlike sensing, the communication receiver does \emph{not} know $X[k,m]$ and hence cannot compute $Z[k,m]$.

\section{PA-Aware Analysis}
\label{sec:pa}
\begin{figure*}[!t]
    \centering
    \includegraphics[width=\textwidth]{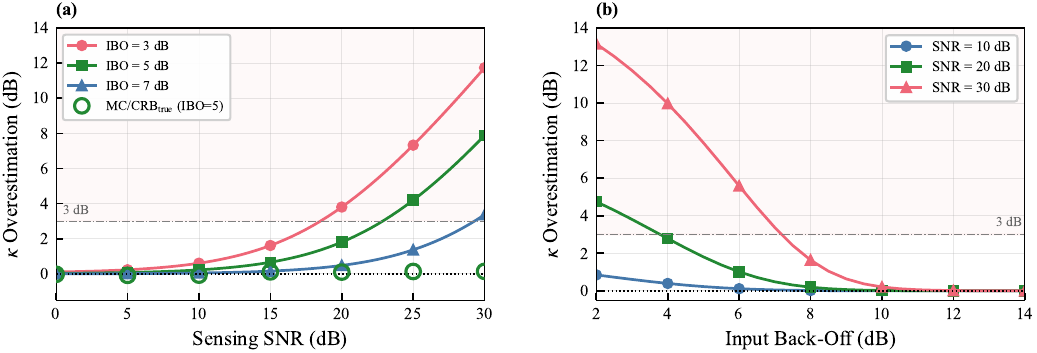}
    \caption{PA sensing CRB analysis. (a)~$\kappa$ model overestimation ratio vs.\ SNR: the physics-based CRB is validated by MC (hollow circles near 0~dB); the $\kappa$ model diverges at high SNR. (b)~Overestimation vs.\ IBO at fixed SNR values.}
    \label{fig:pa}
\end{figure*}

\subsection{Physics-Based Sensing CRB}

We first analyze PA in isolation ($\theta_m = 0$).
Since $Z[k,m]$ is known, the observation~\eqref{eq:obs_sensing} has known mean $\mu_{km} = \alpha_s\, e^{-j2\pi k'\Delta f\tau}\, e^{j2\pi\nu m'T}\, Z[k,m]$ and noise~$\sigma^2$.
Applying the Slepian--Bangs formula~\cite{VanTrees2013} to the target parameter vector $\bm{\eta}_d = [\tau, \nu, \alpha_R, \alpha_I]^T$ yields:

\begin{theorem}[PA-Aware Sensing CRB]
\label{thm:pa_crb}
The CRB for delay and Doppler estimation under PA nonlinearity is
\begin{equation}
    \mathrm{CRB}^{\mathrm{PA}}(\tau) = \frac{\sigma^2}{2|\alpha_s|^2(2\pi\Delta f)^2 \tilde{E}_{20}^{(Z)}},
    \label{eq:crb_pa_tau}
\end{equation}
\begin{equation}
    \mathrm{CRB}^{\mathrm{PA}}(\nu) = \frac{\sigma^2}{2|\alpha_s|^2(2\pi T)^2 \tilde{E}_{02}^{(Z)}},
    \label{eq:crb_pa_nu}
\end{equation}
where $\tilde{E}_{20}^{(Z)} \!= \!\sum_{k,m}(k')^2|Z[k,m]|^2$ and $\tilde{E}_{02}^{(Z)} \!= \!\sum_{k,m}(m')^2|Z[k,m]|^2$ are spectral moments of the PA output.
The degradation relative to the ideal case is $\Delta_{\mathrm{sens}}^{\mathrm{PA}} = 10\log_{10}(\tilde{E}_{20}^{(X)}/\tilde{E}_{20}^{(Z)})$~dB, which is \textbf{SNR-independent}. For the Rapp model with $p=3$ and 16-QAM, $\Delta_{\mathrm{sens}}^{\mathrm{PA}}$ remains below 1.1~dB for IBO~$\geq 3$~dB.
\end{theorem}

\begin{IEEEproof}
Follows from the Slepian--Bangs formula with $Z[k,m]$ replacing $X[k,m]$ as the known waveform.
Under centered indexing and approximately symmetric energy allocation, the off-diagonal coupling terms of the FIM are negligible (verified numerically: residual $< 10^{-4}$), yielding the diagonal CRB~\eqref{eq:crb_pa_tau}--\eqref{eq:crb_pa_nu}.
\end{IEEEproof}

\emph{Physical insight:} The small degradation ($<$1.1~dB) occurs because PA spectral regrowth partially compensates the signal compression---the distortion energy adds useful high-frequency content to $\tilde{E}_{20}^{(Z)}$, partially offsetting the power loss.

\subsection{$\kappa$/Bussgang Overestimation}

The Bussgang decomposition~\cite{Bussgang1952} writes $Z[k,m] = \alpha_B X[k,m] + D[k,m]$, where $\alpha_B \triangleq \E[z\, x^*]/\E[|x|^2]$ is the linear gain and $\sigma_d^2 \triangleq \E[|z - \alpha_B x|^2]$ is the uncorrelated distortion variance, both estimated numerically from the time-domain PA input--output under Gaussian-distributed OFDM signals.
We evaluate the standard aggregate $\kappa$ model as used in~\cite{Salman2025, MateosRamos2024}, which treats $D$ as \emph{unknown noise} in the sensing echo, yielding effective noise $\sigma_{\mathrm{eff}}^2 = |\alpha_s|^2\sigma_d^2 + \sigma^2$.

\begin{theorem}[Overestimation Gap]
\label{thm:gap}
The ratio of the Bussgang CRB to the physics-based CRB satisfies
\begin{equation}
    \frac{\mathrm{CRB}^{B}(\tau)}{\mathrm{CRB}^{\mathrm{PA}}(\tau)} = \underbrace{\frac{\tilde{E}_{20}^{(Z)}}{|\alpha_B|^2\tilde{E}_{20}^{(X)}}}_{\approx\, 1} \cdot \underbrace{\left(1 + \frac{|\alpha_s|^2\sigma_d^2}{\sigma^2}\right)}_{\text{spurious, } \propto\, \gamma_s}\!\!,
    \label{eq:gap}
\end{equation}
where the second factor grows linearly with $\gamma_s$, creating a false sensing floor.
\end{theorem}

\begin{IEEEproof}
Follows from the ratio of~\eqref{eq:crb_pa_tau} to the Bussgang CRB, which substitutes $|\alpha_B|^2\tilde{E}_{20}^{(X)}$ for $\tilde{E}_{20}^{(Z)}$ and uses noise~$\sigma_{\mathrm{eff}}^2$.
\end{IEEEproof}

\subsection{Communication Under PA}

For communication, the Bussgang model is correct: the receiver cannot reconstruct $Z[k,m]$.
The achievable SINR is $\gamma_c^{\mathrm{PA}} = |h|^2|\alpha_B|^2 P_X / (|h|^2\sigma_d^2 + \sigma_c^2)$, and the communication rate is $R = \log_2(1 + \gamma_c^{\mathrm{PA}})$ bits/s/Hz.\footnote{$R$ represents an upper bound assuming Gaussian signaling on a flat-fading channel; achievable rates with finite-alphabet constellations over frequency-selective channels are lower, but the asymmetry conclusions (dB ratios) remain valid.}
The communication degradation $\Delta_{\mathrm{comm}}^{\mathrm{PA}}$ grows with SNR due to a real distortion floor---in stark contrast to the SNR-independent sensing degradation.

\section{PN-Aware Analysis}
\label{sec:pn}
\begin{figure*}[!t]
    \centering
    \includegraphics[width=\textwidth]{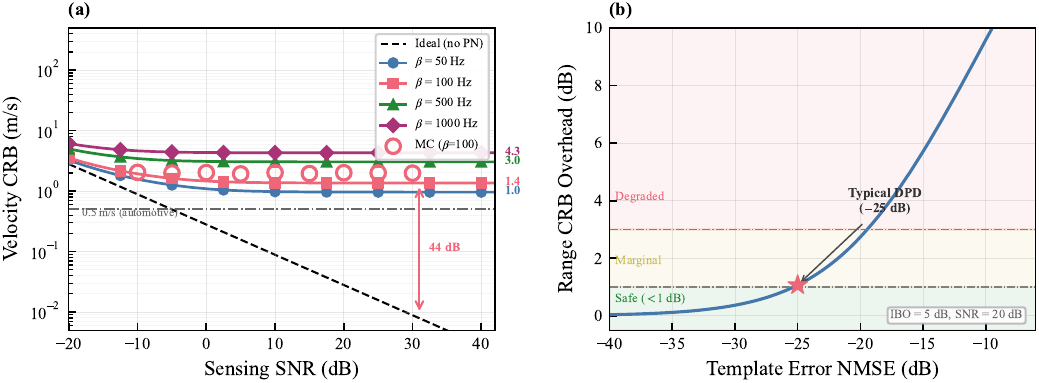}
    \caption{(a)~Velocity CRB vs.\ SNR: L-shaped curves reveal the transition from noise-limited to PN floor-limited regimes. The ideal baseline (dashed) diverges from all PN curves by $>$40~dB at high SNR. MC validates $\beta{=}100$~Hz. (b)~DPD template error sensitivity: range CRB overhead vs.\ template NMSE at IBO~$=5$~dB, SNR~$=20$~dB. Star marks typical DPD accuracy ($-25$~dB).}
    \label{fig:pn}
\end{figure*}
\subsection{Augmented FIM and Schur Complement}

With PN, the parameter vector is augmented: $\bm{\eta} = [\bm{\eta}_d^T, \bm{\theta}^T]^T$.
The random CPE $\bm{\theta}$ has prior $\mathcal{N}(\bm{0}, \bm{C}_\theta)$.
The effective FIM for target parameters follows from the hybrid CRB~\cite{VanTrees2013}:
\begin{equation}
    \bm{J}_{\mathrm{eff}} = \bm{J}_{dd} - \bm{J}_{d\theta}(\bm{J}_{\theta\theta} + \bm{C}_\theta^{-1})^{-1}\bm{J}_{\theta d},
    \label{eq:schur}
\end{equation}
where $\bm{J}_{dd}$, $\bm{J}_{d\theta}$, $\bm{J}_{\theta\theta}$ are blocks of the observation FIM.

\begin{theorem}[PN-Aware Sensing CRB]
\label{thm:pn}
Under CPE phase noise with symmetric spectrum:
\begin{enumerate}
    \item \textbf{Delay} is unaffected: $\mathrm{CRB}^{\mathrm{PN}}(\tau) = \mathrm{CRB}^{\mathrm{ideal}}(\tau)$, since $J_{\tau,\theta_m} \propto \sum_k k'|X[k,m]|^2 = 0$ under symmetric spectrum.
    \item \textbf{Doppler:} The effective FIM element is
    \begin{equation}
        J_{\nu\nu}^{\mathrm{eff}} = \gamma_0(2\pi T)^2\, \bm{m}'^T\!\left(\gamma_0\bm{C}_\theta + \bm{I}_M\right)^{-1}\!\bm{m}',
        \label{eq:J_eff_pn}
    \end{equation}
    where $\gamma_0 = 2|\alpha_s|^2 E_s/\sigma^2 = 2N\gamma_s$ is the per-symbol aggregated SNR, with $E_s = \sum_{k}|X[k,m]|^2 = NP_X$.
    \item \textbf{High-SNR floor:} As $\gamma_0 \to \infty$, $\mathrm{CRB}_\infty^{\mathrm{PN}}(\nu) = [(2\pi T)^2\bm{m}'^T\bm{C}_\theta^{-1}\bm{m}']^{-1}$, an irreducible limit with the velocity floor scaling as $\sqrt{\beta}$.
\end{enumerate}
\end{theorem}

\begin{IEEEproof}
The delay--PN coupling vanishes under symmetric spectrum ($J_{\tau,\theta_m} = 0$).
For Doppler, the Schur complement~\eqref{eq:schur} with $\bm{I} - \gamma_0(\gamma_0\bm{I}+\bm{C}_\theta^{-1})^{-1} = (\gamma_0\bm{C}_\theta+\bm{I})^{-1}$ yields~\eqref{eq:J_eff_pn}.
The floor follows from $(\gamma_0\bm{C}_\theta{+}\bm{I})^{-1} \to (\gamma_0\bm{C}_\theta)^{-1}$ as $\gamma_0 \to \infty$.
\end{IEEEproof}

\emph{Physical insight:} PN adds a random per-symbol phase $\theta_m$ indistinguishable from the Doppler phase $2\pi\nu mT$.
At high SNR, the PN realization is estimated from data, masking the Doppler signature; only the PN prior resolves this ambiguity.

\subsection{Joint PA+PN Separability}

Under the joint observation model~\eqref{eq:obs_sensing} with both PA and PN, the known waveform $Z[k,m]$ replaces $X[k,m]$ in the augmented FIM.
Since $Z[k,m]$ is deterministic, the FIM block structure is preserved:

\begin{corollary}[Orthogonal Design Space]
\label{cor:joint}
Under joint PA nonlinearity and CPE phase noise, the sensing CRB satisfies:
\begin{enumerate}
    \item \textbf{Delay:} $\mathrm{CRB}^{\mathrm{PA{+}PN}}(\tau) = \mathrm{CRB}^{\mathrm{PA}}(\tau)$, determined solely by IBO (PA spectral moments).
    \item \textbf{Doppler floor:} $\mathrm{CRB}_\infty^{\mathrm{PA{+}PN}}(\nu) = [(2\pi T)^2 \bm{m}'^T\bm{C}_\theta^{-1}\bm{m}']^{-1}$, determined solely by $\beta$ (PN statistics), independent of PA.
    \item \textbf{Communication rate:} $R^{\mathrm{PA{+}PN}} \approx R^{\mathrm{PA}}$, determined solely by IBO, since PN contributes $<$0.14~dB loss.
\end{enumerate}
\end{corollary}

\begin{IEEEproof}
Replacing $X[k,m]$ by $Z[k,m]$ in Theorem~\ref{thm:pn}, the per-symbol SNR becomes $\gamma_0^{(Z)} = 2|\alpha_s|^2 E_s^{(Z)}/\sigma^2$, where $E_s^{(Z)} = \sum_{k}|Z[k,m]|^2$.
PA modifies only $E_s^{(Z)}$; the covariance $\bm{C}_\theta$ and symbol-index vector $\bm{m}'$ are unchanged.
In the Doppler floor regime ($\gamma_0^{(Z)} \bm{C}_\theta \gg \bm{I}_M$), the Schur complement~\eqref{eq:J_eff_pn} simplifies to $(2\pi T)^2 \bm{m}'^T \bm{C}_\theta^{-1}\bm{m}'$, where $\gamma_0^{(Z)}$ cancels.
Thus the velocity floor depends exclusively on $\bm{C}_\theta$ (i.e., $\beta$), not on the PA output energy.
For delay, $J_{\tau,\theta_m} \propto \sum_k k'|Z[k,m]|^2 \approx 0$ since PA, being memoryless and applied to real-valued envelopes, preserves the spectral symmetry of $|Z[k,m]|^2$ (verified numerically: residual $< 10^{-4}$ for $NM \geq 3584$).
\end{IEEEproof}

\emph{Design implication:} IBO controls rate (via $\sigma_d^2$) and $\beta$ controls velocity (via $\bm{C}_\theta$); the separability is exact for delay (all SNR) and the Doppler floor (high SNR), with residual coupling vanishing as SNR grows.
Fig.~\ref{fig:design} confirms this: velocity contours are horizontal, rate contours vertical.

\subsection{Communication Under PN}

Communication receivers estimate and remove $\theta_m$ using $N_p$ pilot subcarriers, achieving residual variance $\sigma_{\hat\theta}^2 = 1/(2 N_p \gamma_c)$ at communication SNR~$\gamma_c$~\cite{Wu2004}.
At moderate-to-high SNR ($\gamma_c \gg 1$), the SINR loss converges to $\Delta_{\mathrm{comm}}^{\mathrm{PN}} \approx 10\log_{10}(1 + 1/(2N_p))$~dB, which is below 0.14~dB for $N_p \geq 16$.
Thus, PN is effectively a sensing-only impairment at practical communication SNR.

\section{Numerical Results}
\label{sec:results}

\begin{figure}[!t]
    \centering
    \includegraphics[width=\columnwidth]{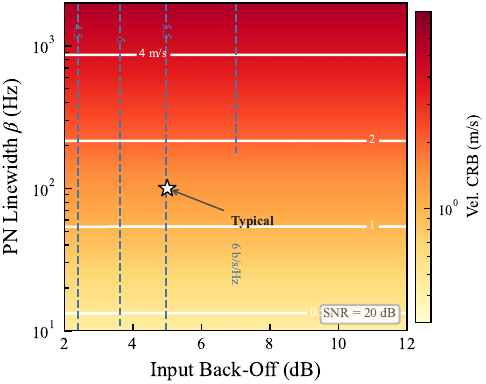}
    \caption{Joint IBO--$\beta$ design map at SNR~$=20$~dB. Color shows velocity CRB (m/s); white contours are velocity thresholds, blue dashed contours are communication rates (bits/s/Hz). Contour orthogonality directly visualizes the separability in Corollary~\ref{cor:joint}. Star marks a typical operating point.}
    \label{fig:design}
\end{figure}

\begin{figure*}[!t]
    \centering
    \includegraphics[width=\textwidth]{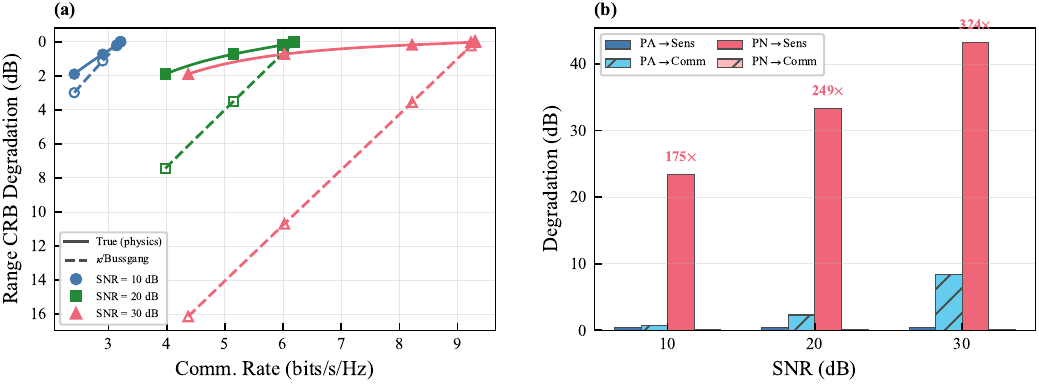}
    \caption{(a)~Sensing--communication Pareto frontier at SNR~$= 10, 20, 30$~dB: the physics-based model (solid) shows significantly smaller sensing cost than the $\kappa$ model (dashed), with the gap widening at higher SNR. (b)~Bidirectional asymmetry (IBO~$=5$~dB, $\beta{=}100$~Hz): PA predominantly degrades communication while PN predominantly degrades sensing. Annotations show the sensing-to-communication degradation-factor ratio $10^{(\Delta_{\mathrm{sens}}-\Delta_{\mathrm{comm}})/10}$.}
    \label{fig:combined}
\end{figure*}
\subsection{PA Overestimation}


The system uses $N{=}256$, $M{=}14$, $\Delta f{=}120$~kHz ($T {\approx} 8.9$~$\mu$s), 16-QAM, $f_c{=}28$~GHz, Rapp PA ($p{=}3$), Wiener CPE, IBO~$=5$~dB, $\beta{=}100$~Hz, and $N_p{=}16$.
MC simulations use grid-search estimation with $N_{\mathrm{MC}}{=}1500$ (PA) and $800$ (PN) realizations.

\subsection{PN Velocity Floor and DPD Robustness}

Fig.~\ref{fig:pa}(a) plots the $\kappa$ model overestimation ratio versus sensing SNR.
MC simulation (hollow circles, IBO~$=5$~dB) validates the physics-based CRB: the MC MSE fluctuates near 0~dB, with deviations attributable to finite-sample variance.
The $\kappa$ overestimation grows from $<$1~dB at low SNR to 12~dB at SNR~$=30$~dB with IBO~$=3$~dB; even at IBO~$=7$~dB, it reaches ${\sim}4$~dB at 30~dB.
Fig.~\ref{fig:pa}(b) confirms that the overestimation exceeds 3~dB for IBO~$\leq 8$~dB at SNR~$=30$~dB, demonstrating that the $\kappa$ model is unreliable across the practical operating region.

Fig.~\ref{fig:pn}(a) plots the absolute velocity CRB versus sensing SNR.
Each curve exhibits an L-shaped transition: at low SNR, the CRB follows the ideal (noise-limited) baseline, while at high SNR it saturates at a PN-induced floor that scales as $\sqrt{\beta}$ (Theorem~\ref{thm:pn}), yielding floors of 0.96, 1.36, 3.04, and 4.30~m/s for $\beta = 50, 100, 500, 1000$~Hz---all exceeding the 0.5~m/s automotive requirement~\cite{Sturm2011}.
MC simulation ($\beta = 100$~Hz) validates the analytical CRB.
The 44~dB gap between the ideal and $\beta{=}100$~Hz curves at SNR~$=30$~dB quantifies the wasted transmit power: beyond the PN-limited regime, increasing SNR yields no velocity improvement.

Fig.~\ref{fig:pn}(b) examines the sensitivity of the range CRB to imperfect PA template knowledge.
The overhead remains below 1~dB for template error NMSE~$\leq -25$~dB---the typical accuracy of practical DPD~\cite{Ghannouchi2009} (star)---and below 0.5~dB for NMSE~$\leq -30$~dB, validating that the ``$Z$ known'' assumption is robust under practical calibration accuracy.

\subsection{Joint IBO--$\beta$ Design Map}

Fig.~\ref{fig:design} maps velocity CRB and communication rate jointly over the (IBO, $\beta$) design space at SNR~$=20$~dB.
The velocity contours (white) are nearly horizontal ($\beta$-dominated) while the rate contours (blue) are nearly vertical (IBO-dominated), directly visualizing the orthogonal separability of Corollary~\ref{cor:joint}.
A system designer can independently specify the PA linearity for communication and the oscillator quality for sensing.
For example, at the typical operating point (star), relaxing PA linearity from IBO~$=5$ to 3~dB reduces the rate by 0.8~bits/s/Hz but changes the velocity CRB by $<$0.001~m/s; improving oscillator quality from $\beta=100$ to 50~Hz halves the velocity CRB without affecting the rate.

\subsection{Pareto Frontier and Asymmetry}

Fig.~\ref{fig:combined}(a) compares the sensing--communication Pareto frontiers under the physics-based and $\kappa$/Bussgang models at SNR~$= 10, 20, 30$~dB.
At each SNR, the true frontier (solid) lies well above the $\kappa$ prediction (dashed), with the gap growing from $<$1~dB at 10~dB to ${\sim}8$~dB at 30~dB---confirming that $\kappa$ model pessimism worsens precisely where accurate sensing analysis matters most.
Fig.~\ref{fig:combined}(b) quantifies the bidirectional asymmetry at IBO~$=5$~dB: at SNR~$=20$~dB, PA degrades communication by 2.3~dB but sensing by only 0.49~dB, while PN degrades sensing by 33.3~dB but communication by only 0.13~dB, a gap of 33.2~dB.

\section{Conclusion}
\label{sec:conclusion}

The $\kappa$ distortion model is inadequate for monostatic sensing in ISAC: a physics-based analysis reveals a bidirectional asymmetry yielding an orthogonal (IBO, $\beta$) design space (Corollary~\ref{cor:joint}), robust to practical DPD errors.
PA linearity is primarily a communication requirement; oscillator quality is primarily a sensing requirement.
Extension to MIMO, memory PA effects, I/Q imbalance, and multi-target scenarios is left for future work.

\balance
\bibliographystyle{IEEEtran}
\bibliography{refs}

\end{document}